\documentclass[11pt]{amsart}
\usepackage{amsaddr}
\usepackage{graphicx,amssymb,amsmath,amsthm}

\usepackage[utf8]{inputenc}
\usepackage{epstopdf}
\usepackage{bm}
\usepackage[numbers]{natbib}
\usepackage{color}
\usepackage{array}

\def\sech{\mathop{\rm sech}\nolimits}

\newcommand{\rem}[1]{}

\newtheorem{lemma}{Lemma}
\newtheorem{remark}{Remark}

\usepackage{amsfonts}

\title{On the impossibility of solitary Rossby waves in meridionally unbounded domains}
\author{Georg A. Gottwald}
\address{School of Mathematics and Statistics, University of Sydney, NSW 2006, Australia}
\email{georg.gottwald@sydney.edu.au} 
\author{Dmitry E. Pelinovsky}
\address{Department of Mathematics, McMaster University, Hamilton ON Canada L8S 4K1   \\
Department of Applied Mathematics, Nizhny Novgorod State Technical University, 603950, Russia}
\email{dmpeli@math.mcmaster.ca}

\begin{document}


\begin{abstract}
Evolution of weakly nonlinear and slowly varying Rossby waves in planetary atmospheres and oceans
is considered within the quasi-geostrophic equation on unbounded domains.
When the mean flow profile has a jump in the ambient potential vorticity, localized eigenmodes are trapped
by the mean flow with a non-resonant speed of propagation. We address amplitude equations for these modes.
Whereas the linear problem is suggestive of a two-dimensional  Zakharov-Kuznetsov equation, we found that
the dynamics of Rossby waves is effectively linear and moreover confined to zonal waveguides of the mean flow.
This eliminates even the ubiquitous Korteweg-de Vries equations
as underlying models for spatially localized coherent structures in these geophysical flows.\\

\vspace{0.1cm}
\centerline{\sc \large dedicated to Roger Grimshaw, our supervisor, to whom we owe so much}
\end{abstract}



\maketitle


\section{Introduction}

Planetary atmospheres and oceans are strongly turbulent media. However, highly ordered coherent structures arise in a process of self-organization, and dominate the dynamics on slow temporal and large spatial scales. Rapidly rotating geophysical flows with small variations in stratification compared to the background stratification are well described by the quasi-geostrophic equation
\begin{align}
\label{charney}
\frac{D q}{Dt}=0,
\end{align}
where $q = \nabla^2\psi + \beta y - F\psi$ is the shallow-water potential vorticity,
$D/Dt=\partial_t+u\partial_x+v\partial_y$ denotes the material derivative,
$\psi$ is the stream function for the horizontal velocities given by $(u,v)=(-\psi_y,\psi_x)$,
$f=f_0+\beta y$ describes the ambient rotation of the planet at latitude $y$,
and $F=1/L_R^2$ is defined in terms of the Rossby radius of deformation $L_R=\sqrt{g\bar H}/f_0$
with Earth's gravitational acceleration $g$ and typical height $\bar H$ \cite{Vallis,Salmon,Pedlosky}.
This equation was first derived by  Charney~\cite{Charney48}, and independently by Obukhov \cite{Obukhov49}. In the context of
low-frequency drift waves in magnetized plasmas the quasi-geostrophic equation
(\ref{charney}) is known as the Hasegawa--Mima equation \cite{HasegawaMima78}.
In the presence of an ambient meridional mean flow $U$ that depends on latitude $y$, we employ the decomposition
$$
\psi = -\int_0^y U(y) dy + \tilde\psi
$$
and $q=q^{(0)}+q^{(1)}$ with
$$
q^{(0)} = \beta y + F \int_0^y U(y) dy - U', \quad q^{(1)} = \nabla^2\tilde\psi-F\tilde\psi,
$$
after which the quasi-geostrophic equation (\ref{charney}) is expressed as
\begin{align}
\label{qgp}
\left(  \frac{\partial}{\partial t} + U\frac{\partial}{\partial x}
\right) \left( \nabla^2 \tilde\psi - F \tilde\psi \right) + (\beta + F U - U'')\, \frac{\partial \tilde\psi}{\partial x} +
J \left( \tilde\psi, \nabla^2 \tilde\psi \right) = 0,
\end{align}
where $J(a,b) = a_xb_y-a_yb_x$ is the Jacobian encapsulating the nonlinearity coming from the material derivative
and the term $\beta+FU-U'' \equiv q^{(0)}_y$ represents the leading order potential vorticity gradient.

The main purpose of this work is to address the underlying evolution equations for the stream function $\tilde\psi(x,y)$ valid on long spatial and
temporal scales in a weakly nonlinear analysis. In this context, one-dimensional solitary Rossby waves were
discussed in \cite{MaxworthyRedekopp76,Redekopp77} where the Korteweg-de Vries (KdV)  equation and the modified KdV
equation were formally derived to describe the persistence of the Great Red Spot in the Jovian atmosphere.
The associated linear problem was stated in these works without much analysis. The work of \cite{MaxworthyRedekopp76,Redekopp77} spawned a huge activity in deriving KdV equations \cite{PatoineWarn82,WarnBrasnett83,HainesMalanotteRizzoli91,MalguzziMalanotteRizzoli84,MalguzziMalanotteRizzoli85,Mitsudera94,GottwaldGrimshaw99,GottwaldGrimshaw99b} and the (bidirectional) Boussinesq equation \cite{HelfrichPedlosky93,HelfrichPedlosky95} in the geophysical context to describe and identify mechanisms for atmospheric blocking, cyclogenesis, meandering of oceanic streams. Most of these works considered laterally bounded domains allowing for one-dimensional propagation in the zonal direction of large-scale solitary waves.

In the present consideration, we investigate the possibility of solitary Rossby waves in the unbounded domain,
and whether one can derive two-dimensional extensions of the KdV models such as the Zakharov-Kuznetsov (ZK)
equation \cite{ZakharovKuznetsov74} which supports stable lump solitary waves. We establish a number of rigorous results on the characterization of
eigenvalues of the associated Rayleigh--Kuo spectral problem, which allow us to characterize
localization in the meridional direction. In particular, we prove under some natural conditions
that no localized eigenmodes with speeds below the wave continuum exist for smooth flow $U$.
However, if the meridional mean flow $U$ has a jump of the ambient potential vorticity,
we show that localized eigenmodes do exist and are trapped by the jump of the mean flow.

A formal derivation of amplitude equations, however, reveals that nonlinear solitary wave equations are not possible in the geophysical situation, neither one- nor two-dimensional. Rather the dynamics of small large-scale localized perturbations is governed by linear dispersion and wave propagation is confined to zonal wave guides prescribed by the linear localized eigenmodes of the associated Rayleigh-Kuo problem. We corroborate this
prediction of the asymptotic analysis by direct numerical simulation of the quasi-geostrophic equation (\ref{qgp}).

The paper is organized as follows. In Section~\ref{sec.LP} we develop the linear theory
for the quasi-geostrophic equation (\ref{qgp}). In Section~\ref{sec.amplitudeeqns}
we present the formal multiple scale analysis demonstrating the impossibility of non-linear solitary waves in unbounded domains. Section~\ref{sec.num} presents numerical simulations of localized initial conditions, which disperse away in the time evolution. Section~\ref{sec.disc} concludes with a discussion. 


\section{Linear theory}
\label{sec.LP}

Before we consider the weakly nonlinear and slowly varying approximation in Section \ref{sec.amplitudeeqns}, 
it is necessary to discuss some properties of the linearised version of the quasi-geostrophic equation (\ref{qgp})
with a non-constant mean flow in terms of normal mode analysis \cite{Pedlosky,Vallis}.
Linearisation of Eq. (\ref{qgp}) yields
\begin{align}
\label{lin1}
\left( \frac{\partial}{\partial t} + U \frac{\partial}{\partial x}
\right) (\nabla^2 \tilde\psi - F\tilde\psi) + (\beta + F U - U'') \frac{\partial \tilde\psi}{\partial x} = 0,
\end{align}
where $U$ depends on $y$ only. Separating variables with the normal mode $\tilde\psi(x,y,t) = e^{i k (x - c t)} \phi(y)$,
where $k \in \mathbb{R}$ is the zonal wave number, $c$ is the phase speed, and $\phi$ is the meridional profile, we obtain the
Rayleigh--Kuo spectral problem
\begin{align}
\label{e-RK}
( U - c) (\phi''-(F+k^2)\phi) + (\beta + F U - U'') \phi = 0,
\end{align}
where $c$ is the spectral parameter and $\phi$ is an eigenfunction to be found.

If $U(y) = \bar{U}$ is a constant mean flow, the spectral problem (\ref{e-RK}) admits only the continuous spectrum located at
\begin{align}
c = \bar{U} - \frac{\beta + F \bar{U}}{F + k^2 + \ell^2},
\label{phase-speed}
\end{align}
where $\ell \in \mathbb{R}$ is the meridional wave number for the Fourier mode $\phi(y) = e^{i \ell y}$.
Expanding (\ref{phase-speed}) in the long-wave limit as
\begin{align}
c = -\frac{\beta}{F} + \frac{\beta + F \bar{U}}{F^2} (k^2 + \ell^2) + \mathcal{O}((k^2 + \ell^2)^2),
\label{e.dispZK}
\end{align}
we obtain the phase speed of the linearized ZK equation \cite{ZakharovKuznetsov74}
for the two-dimensional perturbations on the constant background $\bar{U}$.
However, it is impossible to justify the quadratic nonlinearity of the ZK equation if
we start with the quasi-geostrophic equation (\ref{qgp}) for the constant mean flow $U = \bar{U}$
because the limiting Fourier mode $\phi = 1$ corresponds to $\ell = 0$. This prompts us
to look at the $y$-dependent mean flow $U$ such that $U(y) \to \bar{U}$ as $|y| \to \infty$
and to seek a localized eigenfunction $\phi$ such that $\phi(y) \to 0$ as $|y| \to \infty$
for an eigenvalue $c$ outside the continuous spectrum $[-\beta/F,\bar{U}]$, where
we assume that $\beta, F, \bar{U}$ are all positive. To avoid resonances and critical layers we assume $U(y)+\beta/F>0$ for all $y$.

In particular, we are looking at
the eigenvalues $c$ located below the continuous spectrum, that is, $c < -\beta/F$. 
Although the dispersion relation (\ref{e.dispZK}) suggests that the ZK equation may be the appropriate two-dimensional nonlinear wave model 
for localized perturbations of the $y$-dependent mean flow $U$, we will show in Section~\ref{sec.amplitudeeqns} that this is not the case. To preempt our results, we will see that the amplitude equation contains neither the quadratic nonlinearity nor the meridional component of the dispersion.

In order to formulate rigorous results of the linear theory, we place the Rayleigh--Kuo spectral problem (\ref{e-RK})
in a functional-analytic setting. Because we consider $x$-dependent perturbations in the long-wave limit,
we set $k = 0$ and rewrite the limiting spectral problem in the form
\begin{equation}
\label{Lop}
\mathcal{L}(c) \phi = 0,
\end{equation}
where
\begin{equation*}
\mathcal{L}(c) := (U-c)(\partial_y^2-F) + \beta + F U - U''.
\end{equation*}
If $U, U'' \in L^{\infty}(\mathbb{R})$, then for every $c \in \mathbb{R}$,
$\mathcal{L}(c) : H^2(\mathbb{R}) \subset L^2(\mathbb{R}) \mapsto L^2(\mathbb{R})$
is an unbounded non-selfadjoint operator with bounded coefficients.
Since $F > 0$, the self-adjoint Helmholtz operator $(F-\partial_y^2) : H^2(\mathbb{R}) \subset L^2(\mathbb{R}) \mapsto L^2(\mathbb{R})$
is invertible. Hence, introducing $\varphi := (F-\partial_y^2) \phi$, the limiting
spectral problem (\ref{Lop}) can be formulated as a standard eigenvalue problem
\begin{align}
\mathcal{M} \varphi = c\varphi \quad {\rm{with}}\quad \mathcal{M} := U - (\beta + F U - U'')(F-\partial_y^2)^{-1},
\label{Mop}
\end{align}
where $\mathcal{M} : L^2(\mathbb{R}) \mapsto L^2(\mathbb{R})$
is a bounded non-selfadjoint operator. The adjoint eigenvalue problem
\begin{align}
\mathcal{M}^* \theta = c\theta \quad {\rm{with}}\quad \mathcal{M}^* := U - (F-\partial_y^2)^{-1} (\beta + F U - U'')
\label{MopAdj}
\end{align}
coincides if $(U-c) \theta \in H^2(\mathbb{R})$ with the adjoint spectral problem
\begin{equation}
\label{LopAdj}
\mathcal{L}^*(c) \theta = 0,
\end{equation}
where $\mathcal{L}^*(c) : H^2(\mathbb{R}) \subset L^2(\mathbb{R}) \mapsto L^2(\mathbb{R})$ 
is the adjoint operator to $\mathcal{L}(c)$ given by 
\begin{equation*}
\mathcal{L}^*(c) := (\partial_y^2-F) (U-c) + \beta + F U - U''.
\end{equation*}
The condition $(U-c) \theta \in H^2(\mathbb{R})$  is satisfied for any $\theta \in L^2(\mathbb{R})$ satisfying (\ref{MopAdj})
because if $\theta \in L^2(\mathbb{R})$, then $(F-\partial_y^2)^{-1} (\beta + F U - U'') \in H^2(\mathbb{R})$,
which implies due to the spectral problem $\mathcal{M}^* \theta = c \theta$ 
that $(U-c) \theta \in H^2(\mathbb{R})$. Hence,
the eigenfunction $\theta$ of the adjoint problems for the bounded $\mathcal{M}^*$
and for the unbounded $\mathcal{L}^*(c)$ operators in (\ref{MopAdj}) and (\ref{LopAdj}), respectively, coincide.

The following result is concerned with Fredholm theory for a simple isolated eigenvalue.
\begin{lemma}
\label{lemma-1}
Assume $U, U'' \in L^{\infty}(\mathbb{R})$ and let $c_0 \in \mathbb{R}$ be
a simple isolated eigenvalue of the spectral problem (\ref{Mop})
with the eigenfunction $\varphi_0 \in L^2(\mathbb{R})$. Then,
$c_0$ is also the simple isolated eigenvalue of the adjoint spectral problem (\ref{MopAdj})
with the eigenfunction $\theta_0 \in L^2(\mathbb{R})$
and the inner product $\langle \theta_0, \varphi_0 \rangle_{L^2}$ is nonzero.
\end{lemma}

\begin{proof}
The result follows by Fredholm theory for bounded operators since the reformulation
of the spectral problem (\ref{Lop}) into the form (\ref{Mop}) involves
the bounded operator $\mathcal{M}$.
\end{proof}

\begin{remark}
For each eigenfunction $\varphi_0 \in L^2(\mathbb{R})$ of the spectral problem (\ref{Mop}), one can define
the eigenfunction of the spectral problem (\ref{LopAdj}) by $\phi_0 := (F-\partial_y^2)^{-1} \varphi_0 \in H^2(\mathbb{R})$ 
and rewrite the inner product in the form
\begin{equation}
\label{inner-product}
0 \neq \langle \theta_0, \varphi_0 \rangle_{L^2} = \langle \theta_0, (F-\partial_y^2) \phi_0 \rangle_{L^2}.
\end{equation}
In what follows, we normalize $\theta_0$ from the condition that the inner product (\ref{inner-product}) is one.
\end{remark}

The following result states that no eigenvalues $c$ generally exist below $-\beta/F$
if $U$ is smooth.

\begin{lemma}
\label{lemma-2}
Assume that $U$ is a smooth bounded function of $y$ and
that $U(y) + \beta/F > 0$ for every $y \in \mathbb{R}$. Then,
the spectral problem (\ref{Lop}) admits no eigenvalues $c$ with $c < -\beta/F$.
\end{lemma}

\begin{proof}
Let $c = -\beta/F + \gamma$ and rewrite the spectral problem (\ref{Lop})
in the equivalent form
$$
-(U-c) \phi'' + U'' \phi = \gamma F \phi.
$$
We are looking for eigenvalues $\gamma < 0$ with eigenfunction $\phi \in H^2(\mathbb{R})$.
By using the quotient rule, the spectral problem can be rewritten in the form
$$
- \frac{d}{dy} (U-c)^2 \frac{d}{dy} \frac{\phi}{U-c} = \gamma F \phi,
$$
thanks to the fact that $U(y) - c > 0$ for every $y \in \mathbb{R}$ 
under the assumptions of the lemma. Assume that $\phi \in H^2(\mathbb{R})$
is an eigenfunction for an eigenvalue $\gamma < 0$. 
The existence of this eigenfunction contradicts 
the first Green's identity
$$
\gamma F \int_{\mathbb{R}} \frac{\phi^2}{U-c} dy = \int_{\mathbb{R}} (U-c)^2 \left[ \frac{d}{dy} \frac{\phi}{U-c} \right]^2 dy,
$$
where the right-hand side is positive, hence $\gamma \geq 0$.
\end{proof}

The following result also eliminates possibility of eigenvalues for convex $U$.

\begin{lemma}
\label{lemma-3}
Assume that $U,U'' \in L^{\infty}(\mathbb{R})$ with $U(y) + \beta/F > 0$ and $U''(y) \geq 0$
for every $y \in \mathbb{R}$. Then, the spectral problem (\ref{Lop}) admits no eigenvalues $c$ with $c < -\beta/F$.
\end{lemma}

\begin{proof}
Under assumptions of the lemma, we rewrite the spectral problem (\ref{Lop})
in another equivalent form
$$
- \phi'' = \frac{cF + \beta - U''}{U-c} \phi.
$$
It follows again from the first Green's identity that
$$
\int_{\mathbb{R}} (\phi')^2 dy = \int_{\mathbb{R}} \frac{cF + \beta - U''}{U-c} \phi^2 dy,
$$
where the left-hand side is positive, whereas the right-hand side is negative and well-defined under the assumptions
of the lemma. The contradiction excludes eigenvalues with $c < -\beta/F$.
\end{proof}

Because of the negative results in Lemmas \ref{lemma-2} and \ref{lemma-3}, we have to consider
piecewise smooth configuration $U$. In particular, we consider examples of symmetric localized jets on a constant mean flow of the form
\begin{align}
U(y) = \bar U \left[ 1- a \exp(-b |y|) \right], \quad y \in \mathbb{R},
\label{e.Usym}
\end{align}
with positive parameters $\bar{U}$, $a$, and $b$,
and asymmetric localized jets of the form
\begin{align}
U(y)=
\begin{cases}
\bar U_- \left[ 1- a_- \exp(b_- y) \right] & y < 0,\\
\bar U_+ \left[ 1- a_+ \exp(-b_+ y) \right] & y > 0,
\end{cases}
\label{e.Uasym}
\end{align}
where $\bar{U}_{\pm}$, $a_{\pm}$, and $b_{\pm}$ are positive parameters. To assure the continuity of
$U$ and $U''$ at $y=0$ for (\ref{e.Uasym}) we require $\bar{U}_- (1-a_-) = \bar{U}_+ (1-a_+)$ and $b_-=b_+\sqrt{a_-U_-/a_+U_+}$.

Both the symmetric and asymmetric flows (\ref{e.Usym}) and (\ref{e.Uasym})
exhibit a discontinuity of their first derivative at $y=0$, which causes the leading order potential vorticity gradient
in $q^{(0)}_y= \beta + F U - U''$ to have a $\delta$-function singularity. It is standard to consider the linear problem
for the regions $y<0$ and $y>0$ separately and to treat $y=0$ as a boundary, see 
\cite[Chapter 9.3.3]{Buhler} and \cite[Chapter 9.2]{Vallis}. After removal of 
the $\delta$-function singularity, $U''$ is continuous across the point $y = 0$. For such mean flows with a jump in the ambient potential vorticity at $y=0$,
the conditions of Lemma \ref{lemma-2} and \ref{lemma-3} are not satisfied and
eigenvalues $c$ of the spectral problem (\ref{Lop}) below the wave continuum may exist with $c < -\beta/F$.

An example for a symmetric mean flow of the form (\ref{e.Usym})
with $\bar U = 5$, $a=0.7$ and $b=0.8$ is shown in the left panel of Figure~\ref{fig.Usym}.
There exists a simple isolated eigenvalue $c_0 \approx -1.73$ for $\beta=1$ and $F=1$ with the eigenfunction $\phi_0$ shown in the right panel.
The left panel of Figure~\ref{fig.Uasym} presents an example for an asymmetric mean flow of the form (\ref{e.Uasym})
with $\bar U_-=2$, $a_-=0.3$, $b_-=0.4$ and $a_+=0.8$, implying $\bar U_+=7$ and $b_+=0.13$. The simple isolated
eigenvalue is located at $c_0 \approx -1.03$ for $\beta=1$ and $F=1$ with the eigenfunction $\phi_0$ shown in the right panel.


%
\begin{figure}[h]
	\begin{tabular}{ll}
		\includegraphics[width = 0.45\columnwidth]{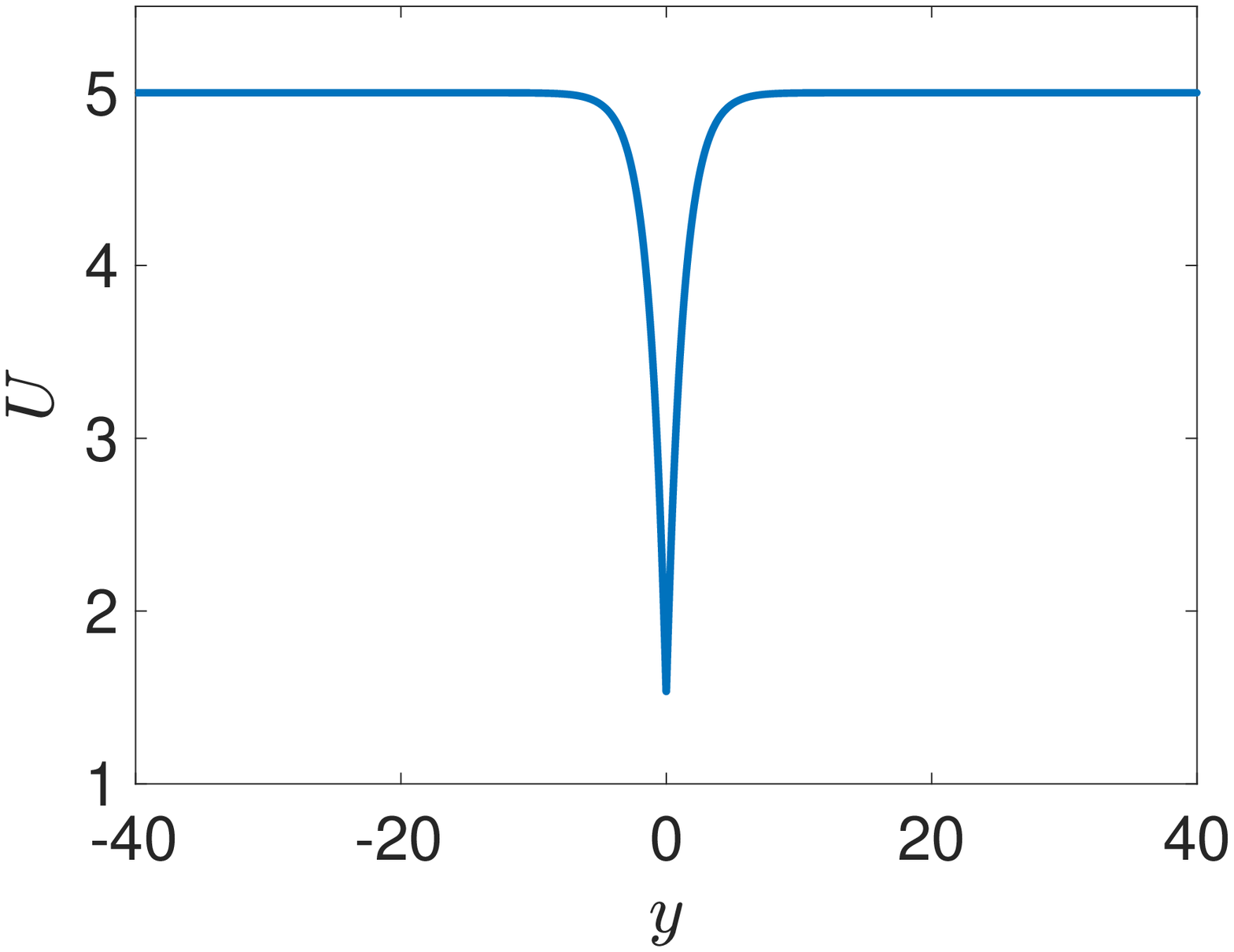}
		\includegraphics[width = 0.45\columnwidth]{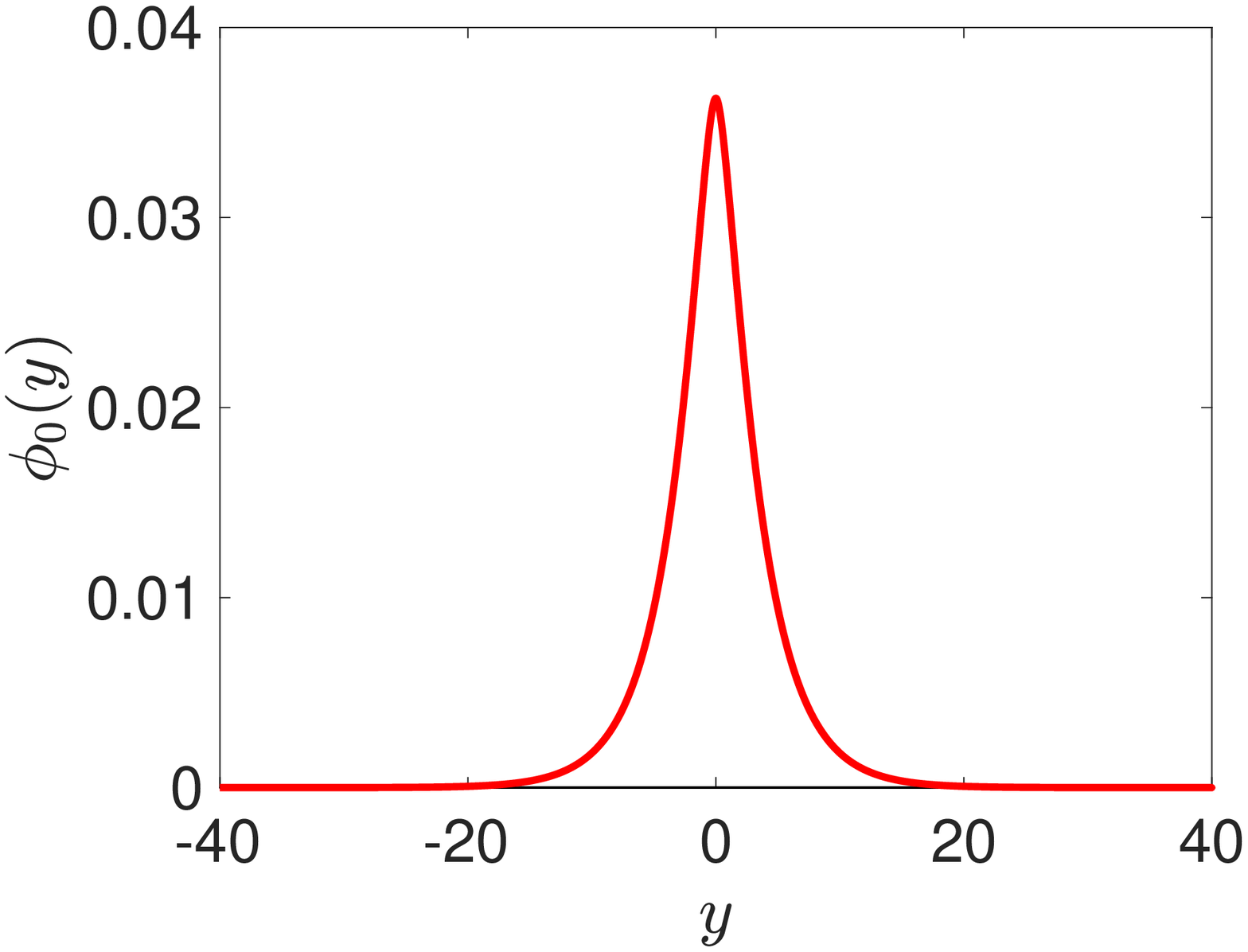}
	\end{tabular}
	\caption{Symmetric mean flow profile $U$ in (\ref{e.Usym}) with $\bar U = 5$, $a=0.7$ and $b=0.8$ for $\beta=1$ and $F=1$ (left)
and the localized eigenfunction $\phi_0$ (right) for the smallest eigenvalue $c_0 \approx -1.73$.}
	\label{fig.Usym}
\end{figure}

\begin{figure}[h]
	\begin{tabular}{ll}
		\includegraphics[width = 0.45\columnwidth]{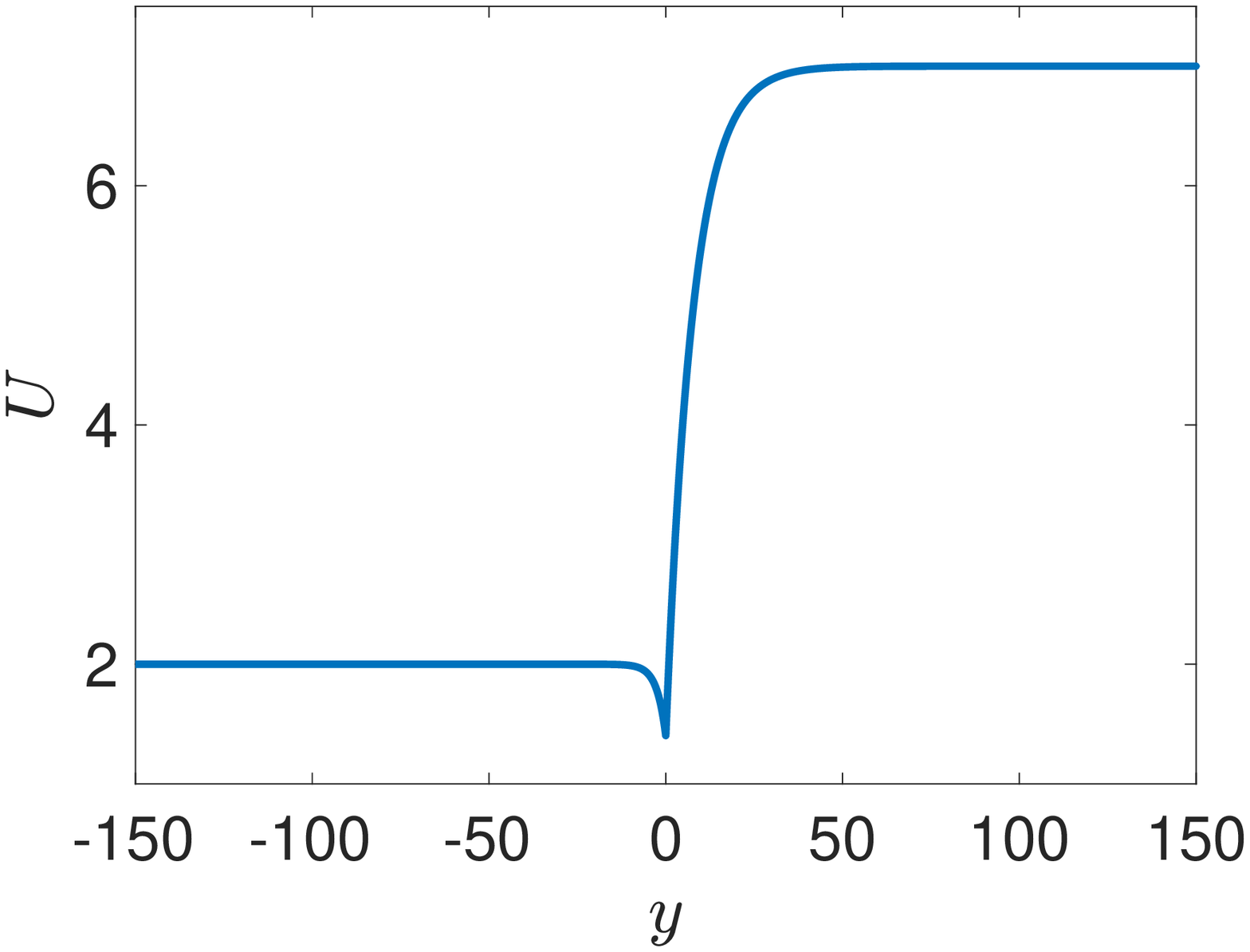}
		\includegraphics[width = 0.45\columnwidth]{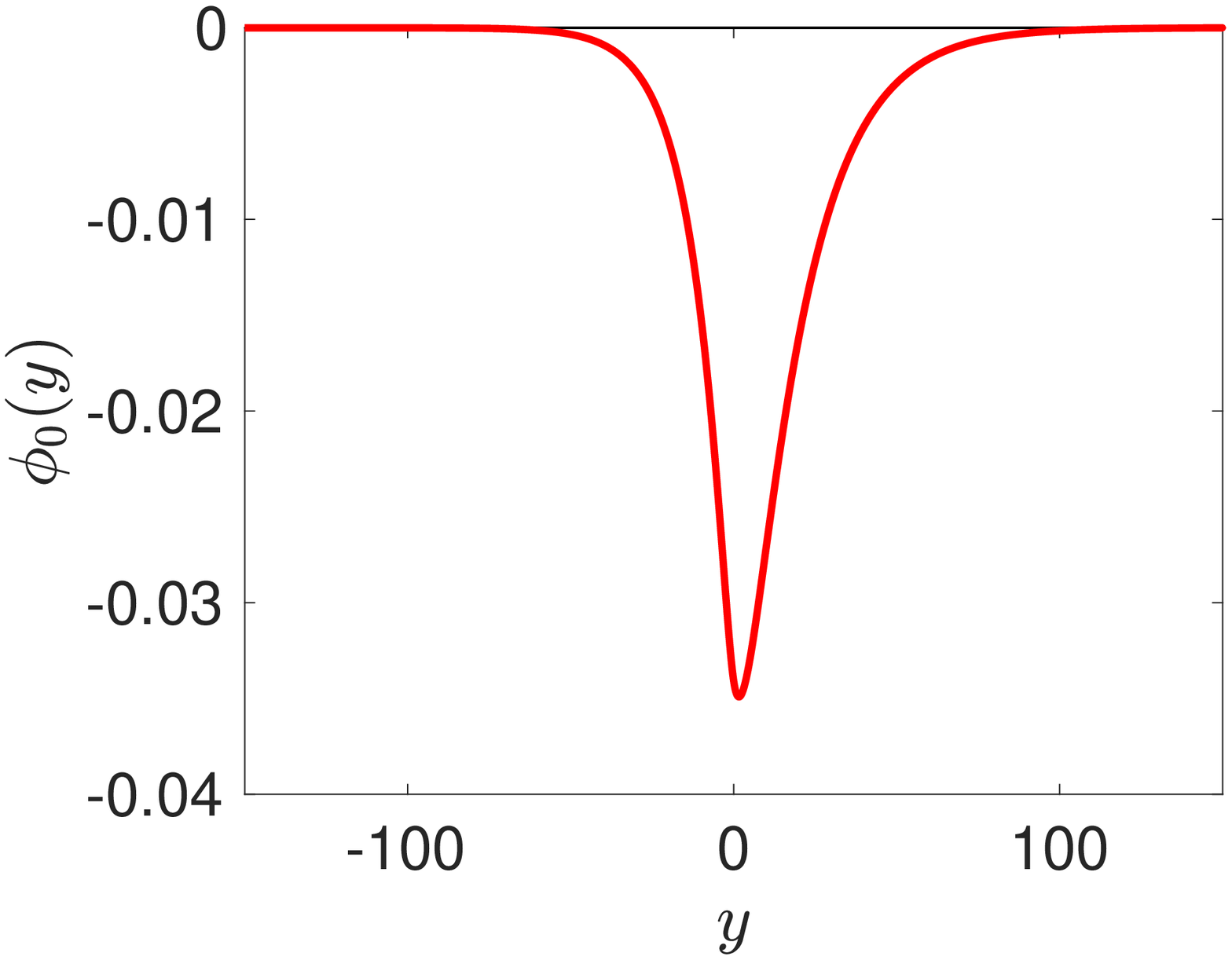}
	\end{tabular}
	\caption{Asymmetric mean flow profile $U$ in (\ref{e.Uasym}) with $\bar U_-=2$, $a_-=0.3$, $b_-=0.4$ and $\bar U_+=7$, $a_+=0.8$, $b_+=0.13$ for $\beta=1$ and $F=1$ (left) and the localized eigenfunction $\phi_0$ (right) for the smallest eigenvalue $c_0 \approx -1.03$.}
	\label{fig.Uasym}
\end{figure}


\section{Derivation of solitary Rossby wave equations}
\label{sec.amplitudeeqns}

Here we derive an evolution equation describing the dynamics of localized pulses
with amplitude $A(X,Y,T)$ on large spatial scales $X,Y$ and long time scales $T$
in the weakly nonlinear and slowly varying reduction of the barotropic equation (\ref{qgp}) 
in an unbounded domain.
We show in Section \ref{sec.KdV} that the classical KdV and ZK equations cannot be derived
and the dynamics of the amplitude $A$ is governed by the linear wave equation
\begin{align}
\label{A_lin}
A_T = \lambda \, A_{XXX},
\end{align}
where $\lambda$ is a numerical constant. We then show in Section \ref{sec.mKdV}
that the dynamics remains linear according to (\ref{A_lin}) even when including higher-order cubic nonlinear terms.


\subsection{KdV and ZK equations}
\label{sec.KdV}

In addition to the fast meridional variable $y$ over which the mean flow $U$ changes,
we introduce the long spatial and slow temporal scales:
\begin{align}
\label{tscaling}
X = \epsilon (x - c_0 t),\hspace{0.7cm}
Y = \epsilon y,\hspace{0.7cm}
T = \epsilon^3 t,\hspace{0.7cm}
\end{align}
where the limiting speed $c_0$ coincides with a simple isolated eigenvalue
of the spectral problem (\ref{Lop}). Derivative terms in the quasi-geostrophic potential vorticity equation (\ref{qgp}) are expanded as follows
\begin{eqnarray*}
\partial_t + U \partial_x = \epsilon (U-c_0) \partial_X + \epsilon^3 \partial_T, 
\end{eqnarray*}
\begin{eqnarray*}
\nabla^2 = \partial_y^2 + 2 \epsilon \partial_y \partial_Y + \epsilon^2 (\partial_X^2 + \partial_Y^2),
\end{eqnarray*}
and
\begin{eqnarray*}
J(a,b) = \epsilon (a_X b_y - a_y b_X) + \epsilon^2 (a_X b_Y - a_Y b_X).
\end{eqnarray*}
We seek an asymptotic solution of the form
\begin{align}
\tilde\psi = \epsilon^2 \psi^{(2)} + \epsilon^3\psi^{(3)} + \epsilon^4  \psi^{(4)} + \cdots, \quad
\label{asympt-exp}
\end{align}
with the leading-order perturbation stream function
\begin{align}
\label{psi-2}
\psi^{(2)} = A(\epsilon(x-c_0t),\epsilon y,\epsilon^3t) \, \phi_0(y),
\end{align}
where $\phi_0 \in H^2(\mathbb{R})$ is the eigenfunction of the spectral problem (\ref{Lop}) for a simple isolated
eigenvalue $c_0$ such that $\mathcal{L}(c_0) \phi_0 = 0$. For the asymptotic expansion (\ref{asympt-exp}) to be a solution of the quasi-geostrophic potential vorticity equation (\ref{qgp}), each term in the asymptotic expansion (\ref{asympt-exp}) must  decay to zero as $|y| \to \infty$. By substituting (\ref{asympt-exp}) into (\ref{qgp}) and using the expansions
in terms of slow variables (\ref{tscaling}), we obtain a sequence of equations at orders of $\mathcal{O}(\epsilon^k)$ with $k \geq 3$.
The choice of $\psi^{(2)}$ in (\ref{psi-2}) satisfies the equation at $\mathcal{O}(\epsilon^3)$.
At the next order, $\mathcal{O}(\epsilon^4)$, we obtain the linear inhomogeneous equation
\begin{align}
\label{psiin1ZK-ZK}
  \mathcal{L}(c_0) \partial_X \psi^{(3)} = - 2(U-c_0) A_{XY} \phi_0^\prime,
\end{align}
which can be solved explicitly to yield
\begin{align}
\psi^{(3)} = -y\phi_0(y) A_Y(X,Y,T).
\label{e-psi1ZK2-ZK}
\end{align}
At the order $\mathcal{O}(\epsilon^5)$ we obtain the linear inhomogeneous equation
\begin{align}
\mathcal{L}(c_0) \partial_X \psi^{(4)} &=
\partial_T A (F - \partial_y^2) \phi_0 + 2(U-c_0)\left(y\phi_0\right)^\prime A_{XYY}
-(A_{XXX}+A_{XYY})(U-c_0)\phi_0 \nonumber \\
&- AA_X (\phi_0\phi_0^{\prime\prime\prime}-\phi_0^\prime\phi_0^{\prime\prime}).
\label{e-eps2ZK-ZK}
\end{align}
Let $\theta_0 \in L^2(\mathbb{R})$ be the eigenvector of the adjoint problem (\ref{LopAdj})
for the same eigenvalue $c_0$ such that $(U-c_0) \theta_0 \in H^2(\mathbb{R})$.
By Lemma \ref{lemma-1}, if $c_0$ is a simple eigenvalue, then $\langle \theta_0, (F - \partial_y^2) \phi_0 \rangle$ is nonzero
and we normalize $\theta_0$ such that $\langle \theta_0, (F - \partial_y^2) \phi_0 \rangle = 1$.
By Fredholm theory, there exists a solution $\psi^{(4)}$ to the linear inhomogeneous equation (\ref{e-eps2ZK-ZK})
decaying to zero as $|y| \to \infty$ if and only if the right-hand side of this equation
is orthogonal to $\theta_0$ in $L^2(\mathbb{R})$. This solvability condition yields
the evolution equation for the amplitude $A(X,Y,T)$ which is given by the following ZK equation
\begin{align}
\label{ZK11}
A_T = \mu A A_{X} + \lambda A_{XXX} + \zeta A_{XYY},
\end{align}
with the numerical coefficients given by the following inner products
\begin{align}
\label{parametersZKmu}
\mu &=\langle \theta_0, (\phi_0 \phi_0^{\prime\prime\prime}- \phi_0^\prime \phi_0^{\prime\prime}) \rangle, \\
\label{parametersZKlambda}
\lambda &= \langle \theta_0, (U-c_0)\phi_0 \rangle, \\
\label{parametersZKzeta}
\zeta &= \langle \theta_0, (U-c_0) (\phi_0 - 2 \left(y\phi_0\right)^\prime) \rangle.
\end{align}

We now show that $\zeta=0$, hence the ZK equation cannot be derived as a two-dimensional extension to the KdV equation.
By using $\mathcal{L}(c_0)(y\phi_0)=2\left(U-c_0\right) \phi_0'$, we write (\ref{parametersZKzeta}) in the form
\begin{align*}
\zeta &= - \langle \theta_0, (U-c_0) (\phi_0 + 2 y\phi_0^\prime) \rangle\\
&=- \langle \theta_0, (U-c_0) \phi_0 \rangle - \langle y\theta_0, \mathcal{L}(c_0)(y\phi_0)\rangle.
\end{align*}
Using $\mathcal{L}^\star(c_0)(y\theta_0)=2\partial_y\left( \left(U-c_0\right)\theta_0\right)$ we obtain after partial integration
that $\zeta = -\zeta$ implying $\zeta=0$.

We next show that in addition $\mu=0$ which precludes the role of the KdV equation to describe localized large scale perturbations.
For symmetric mean flow profiles $U$, the linear problems (\ref{Lop}) and (\ref{LopAdj})  support even eigenfunctions
$\phi_0$ and $\theta_0$. This implies that the integrand in (\ref{parametersZKmu}) is odd so that $\mu=0$ for
symmetric flow profiles. For asymmetric mean flow profiles, we only have numerical evidence for $\mu=0$.
In Figure~\ref{fig.mu} we plot $|\mu|$ versus the number of spatial grid points $N$ and show that 
$\mu\to 0$ as $N \to \infty$ for an asymmetric flow profile on Figure~\ref{fig.Uasym} with $|\mu|\sim1/\sqrt{N}$.

Hence the dynamics of the amplitude $A(X,Y,T)$ 
is entirely linear with zonal dispersion only, and is described by the linear dispersive wave equation (\ref{A_lin}) 
with $\lambda$ given by (\ref{parametersZKlambda}). For the mean flow depicted in Figure~\ref{fig.Usym} we obtain $\lambda \approx -4.69$. 
For the asymmetric mean flow depicted in Figure~\ref{fig.Uasym} we obtain $\lambda=-4.16$.

\begin{figure}[h]
\begin{center}
		\includegraphics[width = 0.45\columnwidth]{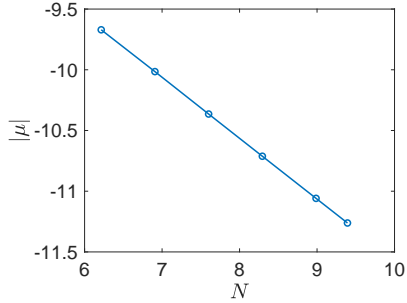}
		\end{center}
	\caption{Log-log plot of $|\mu|$ for increasing resolution with number of spatial grid points $N$ for the asymmetric mean flow depicted in Figure~\ref{fig.Uasym}. The slope is estimated as $-0.5$ from linear regression.}
	\label{fig.mu}
\end{figure}


\subsection{Modified KdV and ZK equations}
\label{sec.mKdV}

Since $\mu = 0$ in the classical ZK equation (\ref{ZK11}), 
we should redefine the scaling of the asymptotic expansion (\ref{asympt-exp}) and derive a modified ZK equation
with cubic nonlinear terms, analogously to \cite{Redekopp77} for the case of the KdV and modified KdV equations. 
We consider the same scaling (\ref{tscaling})
of the slow variables and redefine the asymptotic expansion in the form
\begin{align}
\label{asympt-exp-mod}
\tilde\psi = \epsilon \psi^{(1)} + \epsilon^2\psi^{(2)} + \epsilon^3  \psi^{(3)}  + \cdots,
\end{align}
with the leading-order perturbation stream function
\begin{align}
\label{psi-1}
\psi^{(1)} = A(\epsilon(x-c_0t),\epsilon y, \epsilon^3t) \, \phi_0(y),
\end{align}
where $(c_0,\phi_0)$ are the same as in (\ref{psi-2}) and the corrections of the asymptotic expansion (\ref{asympt-exp-mod}) are still
supposed to decay to zero as $|y| \to \infty$. By substituting (\ref{asympt-exp-mod}) into the quasi-geostrophic potential vorticity equation (\ref{qgp}) and using the expansions
in terms of the slow variables (\ref{tscaling}), we obtain a sequence of equations at orders of $\mathcal{O}(\epsilon^k)$ with $k \geq 2$.
The choice of $\psi^{(1)}$ in (\ref{psi-1}) satisfies the equation at $\mathcal{O}(\epsilon^2)$.
At the next order, $\mathcal{O}(\epsilon^3)$, we obtain the linear inhomogeneous equation
\begin{align}
\label{psiin1mZK}
  \mathcal{L}(c_0) \partial_X \psi^{(2)} = - 2(U-c_0) A_{XY} \phi_0^\prime - A A_X (\phi_0\phi_0^{\prime\prime\prime}-\phi_0^\prime\phi_0^{\prime\prime}).
\end{align}
The explicit solution in (\ref{e-psi1ZK2-ZK}) would satisfy (\ref{psiin1mZK}) if only the first term in the right-hand side
were present. However, we are not able to find an explicit solution
for the full linear equation (\ref{psiin1mZK}). Therefore, we represent
the solution formally as
\begin{align}
\psi^{(2)} = -y\phi_0(y) A_Y(X,Y,T) - \frac{1}{2} \phi_2(y) A(X,Y,T)^2,
\label{e-psi1mZK2}
\end{align}
where $\phi_2$ is a solution of the inhomogeneous equation
\begin{align}
\mathcal{L}(c_0) \phi_2 = \phi_0 \phi_0^{\prime\prime\prime}- \phi_0^\prime \phi_0^{\prime\prime}.
\label{e-psi1mZK2}
\end{align}
A solution $\phi_2 \in H^2(\mathbb{R})$ to this equation exists by Fredholm theory thanks to the constraint $\mu = 0$. 
To make the solution unique, we introduce the
orthogonality condition $\langle \phi_0, \phi_2 \rangle = 0$ on the correction $\phi_2$.

At the order $\mathcal{O}(\epsilon^4)$ we obtain the linear inhomogeneous equation
\begin{align}
\mathcal{L}(c_0) \partial_X \psi^{(3)} &=
\partial_T A (F - \partial_y^2) \phi_0 + 2(U-c_0)\left(y\phi_0\right)^\prime A_{XYY}
-(A_{XXX}+A_{XYY})(U-c_0)\phi_0 \nonumber \\
&+ \partial_X (AA_Y) \left[ 2 (U-c_0) \phi_2^\prime
+ y\left( \phi_0\phi_0^{\prime\prime\prime}-\phi_0^\prime\phi_0^{\prime\prime} \right) \right] \nonumber \\
&+\frac{1}{2}A^2A_X
\left[ \phi_0\phi_2^{\prime\prime\prime}  +2 \phi_0^{\prime\prime\prime} \phi_2 - 2\phi_0^\prime\phi_2^{\prime\prime} - \phi_0^{\prime\prime}\phi_2^\prime
\right].
\label{e-eps2mZK}
\end{align}
As in Section \ref{sec.KdV}, the solvability condition for $\psi^{(3)}$ in (\ref{e-eps2mZK})
yields the evolution equation for the amplitude $A(X,Y,T)$ which is given by the following
modified ZK equation
\begin{align}
\label{ZK11-mod}
A_{T} = \kappa (AA_Y)_X + \nu A^2 A_X + \lambda A_{XXX} + \zeta A_{XYY},
\end{align}
where $\lambda$ and $\zeta$ are the same as in (\ref{parametersZKlambda}) and (\ref{parametersZKzeta}), and $\kappa$ and $\nu$ are defined by
\begin{align}
\label{parametersZK-mod}
\kappa &=- \langle \theta_0, \left[ 2 (U-c_0) \phi_2^\prime + y\left(\phi_0 \phi_0^{\prime\prime\prime}- \phi_0^\prime \phi_0^{\prime\prime}\right)
\right] \rangle, \\
\nu &= -\frac{1}{2} \langle \theta_0, \left[ \phi_0\phi_2^{\prime\prime\prime}  +2 \phi_0^{\prime\prime\prime} \phi_2
- 2\phi_0^\prime\phi_2^{\prime\prime} - \phi_0^{\prime\prime}\phi_2^\prime \right] \rangle.
\end{align}

As we now show, $\kappa=\nu=0$, and hence the amplitude equation is given again by the linear dispersive wave equation (\ref{A_lin}) (recall from Section \ref{sec.KdV} that $\zeta=0$). It is readily seen that $\kappa=0$ since it follows from (\ref{e-psi1mZK2}) that 
\begin{align*}
\kappa &=- 2\langle \theta_0, (U-c_0) \phi_2^\prime\rangle - \langle \mathcal{L}^\star(c_0)(y\theta_0),\phi_2\rangle\\
&=- 2\langle \theta_0, (U-c_0) \phi_2^\prime\rangle - 2 \langle \partial_y\left( \left(U-c_0\right)\theta_0\right),\phi_2\rangle=0.
\end{align*}
We show next that $\nu = 0$ as well. Recall that since $\phi_0^{\prime\prime\prime}$ is discontinuous at $y=0$ so is $\phi_2^{\prime\prime}$, and hence $\phi_2^{\prime\prime\prime}$ involves a $\delta$-function singularity. We therefore perform partial integration to allow for a computationally feasible expression of $\nu$ which requires splitting the integration into integration over $y<0$ and $y>0$. We obtain
\begin{align}
\label{nu_num}
\nu = -\frac{1}{2} \langle \theta_0, \left[ 2 \phi_0^{\prime\prime\prime} \phi_2
- 2\phi_0^\prime\phi_2^{\prime\prime} - \phi_0^{\prime\prime}\phi_2^\prime \right] \rangle +
\frac{1}{2} \langle \theta_0^\prime,  \phi_0\phi_2^{\prime\prime}\rangle +  \frac{1}{2} \langle \theta_0,  \phi_0^\prime\phi_2^{\prime\prime}\rangle -  \frac{1}{2} \left[\theta_0\phi_0\phi_2^{\prime\prime}\right]_{0-\epsilon}^{0+\epsilon},
\end{align}
for $\epsilon\to 0$. 
We were not able to show analytically that $\nu=0$ but have performed
careful numerical experiments for several mean flow configurations confirming $\nu=0$.
In Figure~\ref{fig.nu} we show $\nu$ versus the number of spatial grid points $N$ 
suggesting convergence $\nu\to 0$ as $N \to \infty$ with $\nu\sim 1/N$. 

\begin{figure}[h]
\begin{center}
		\includegraphics[width = 0.45\columnwidth]{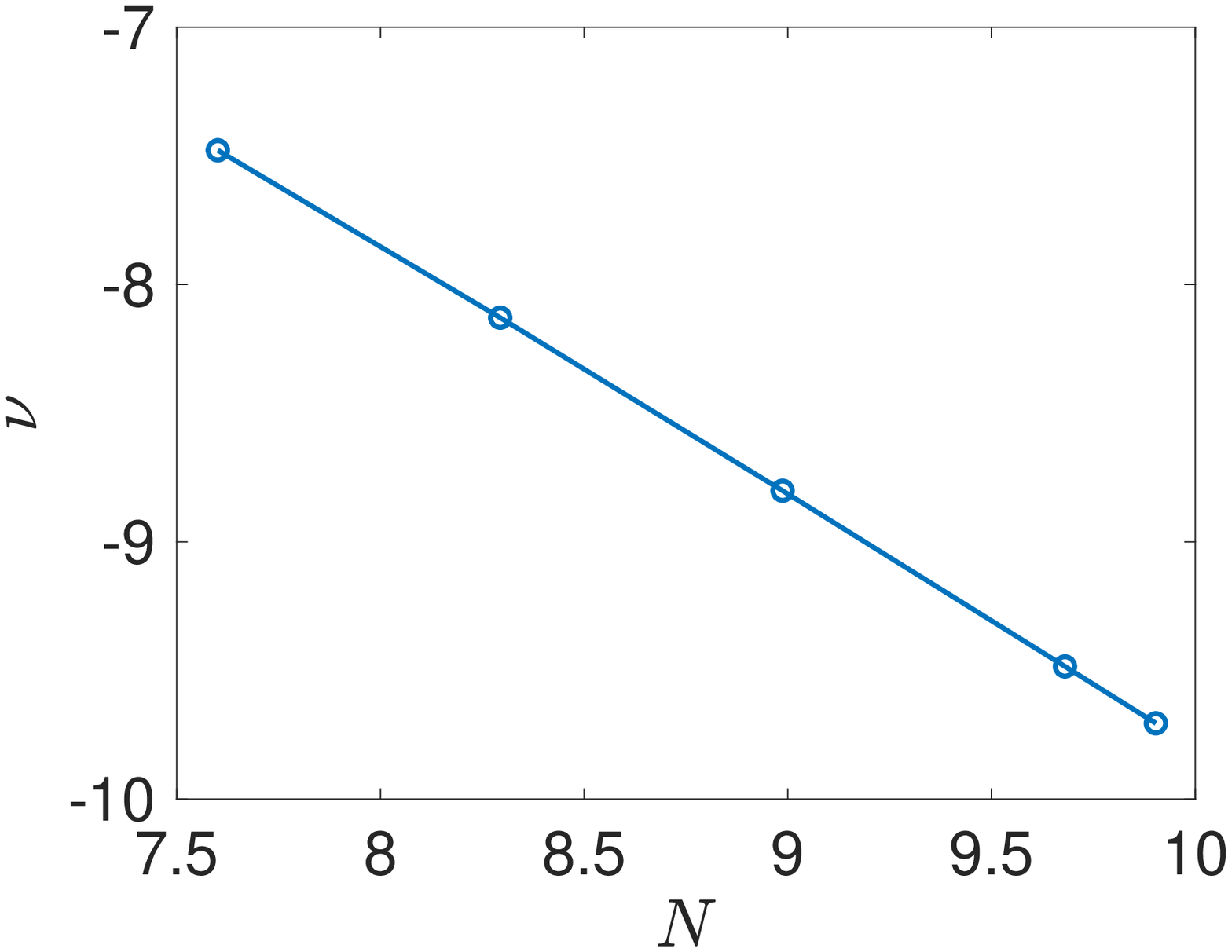}
		\includegraphics[width = 0.45\columnwidth]{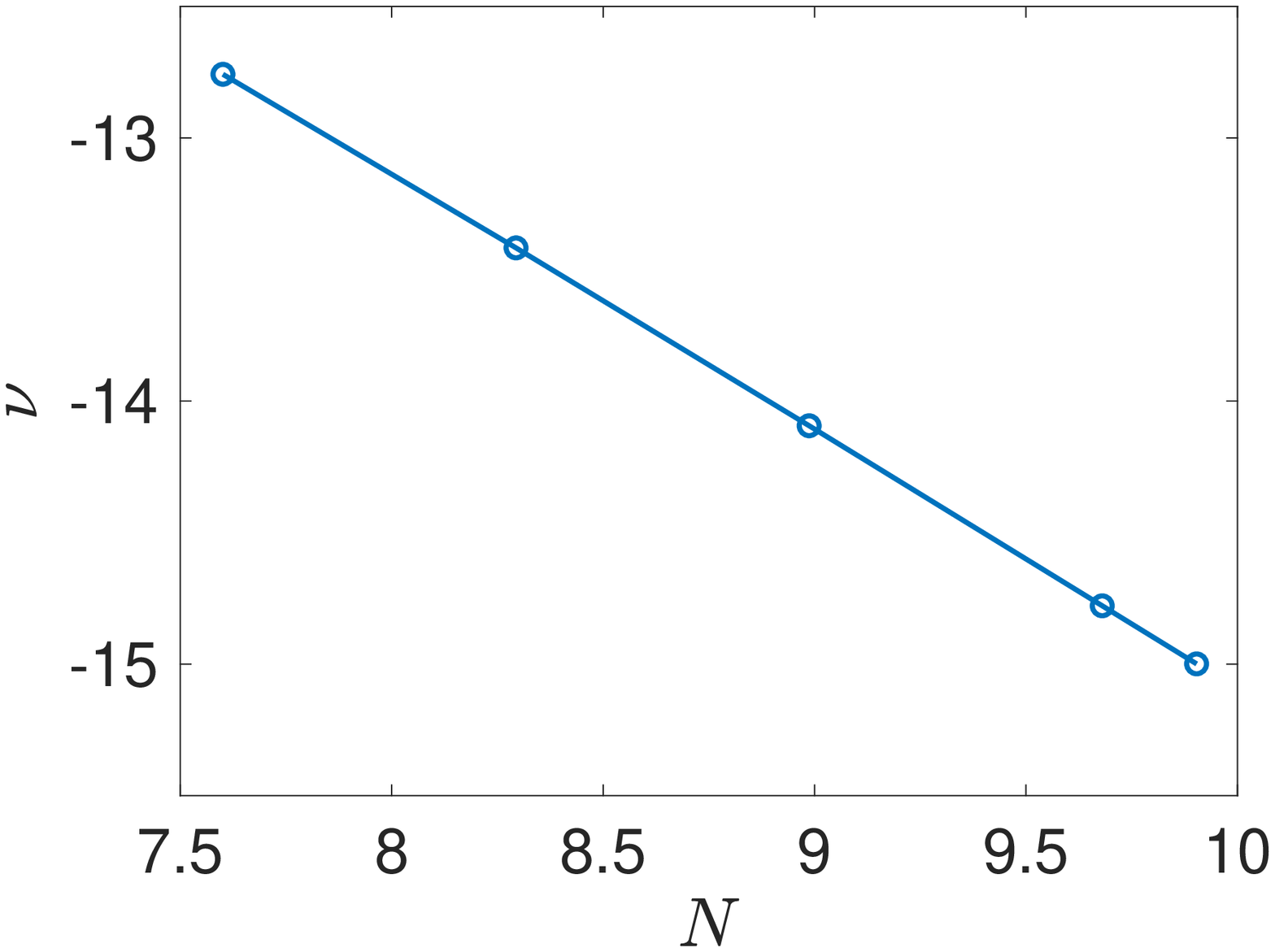}
		\end{center}
	\caption{Log-log plot of $\nu$, numerically calculated from (\ref{nu_num}), for increasing resolution with number of spatial grid points $N$. Left:  symmetric mean flow depicted in Figure~\ref{fig.Usym}. Right: Asymmetric mean flow depicted in Figure~\ref{fig.Uasym}. The slope is estimated as $-0.97$ from linear regression in both cases.}
	\label{fig.nu}
\end{figure}


\section{Numerical simulation of the quasi-geostrophic equations}
\label{sec.num}
Here we numerically integrate the quasi-geostrophic potential vorticity equation (\ref{qgp}) for localized initial conditions of the form
\begin{align}
\label{e.A0}
\tilde \psi(x,y,t=0) = A_0\sech^2(w x)\phi_0(y),
\end{align}
where $\phi_0$ is given as the normalized eigenfunction of the linear problem (\ref{Lop}) corresponding to the 
 isolated eigenvalue $c_0$. We shall use the symmetric mean flow profile (\ref{e.Usym}) with parameters as in Figure~\ref{fig.Usym} as well as the asymmetric mean flow profile (\ref{e.Uasym}) with parameters as in Figure~\ref{fig.Uasym}. We choose $A_0=0.11$ and $w=0.1$ for both mean flow profiles.

Numerical integration is based on the finite-difference scheme where the evolution  problem is split into firstly determing $\tilde\psi$ by solving the Helmholtz problem $(\nabla^2-F)\tilde\psi=q$ for given potential vorticity $q$ in spectral space, and then, in a second step, advecting the potential vorticity in time using a second-order leapfrog scheme \cite{Holland78}. The discretization of the nonlinearity is performed with the Arakawa scheme \cite{Arakawa66} which conserves energy and enstrophy. We choose periodic boundary conditions on a large domain to mimic an unbounded domain. We choose a time step of $\Delta t=0.01$ and a spatial discretization of $\Delta x=0.3$ with a square domain of length $L=900$ throughout. As before we choose $\beta=F=1$.

Figures~\ref{fig.QGUsym} and \ref{fig.QGaUsym} show snapshots of the stream function $\tilde\psi$ and 
 its cross-section at the latitude of the discontinuity in the mean flow $y=0$. It is clearly seen that the initially localized perturbation (\ref{e.A0}) linearly disperses along the zonal direction centered at $y=0$. We have tested this behaviour for several initial conditions as well as for several mean flow profiles.

\begin{figure}[h]
	\begin{tabular}{ll}
		\includegraphics[width = 0.45\columnwidth]{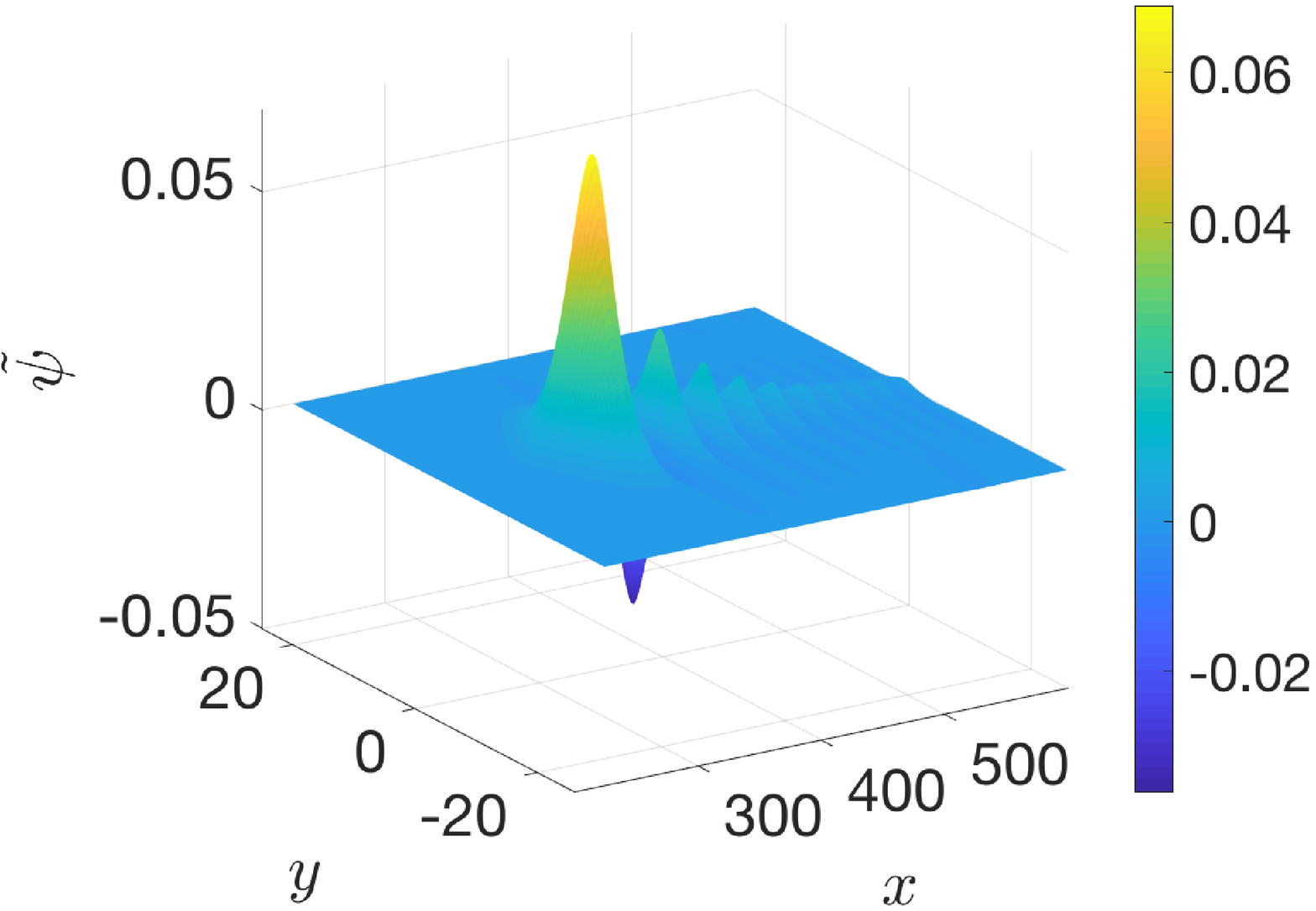}
		\includegraphics[width = 0.45\columnwidth]{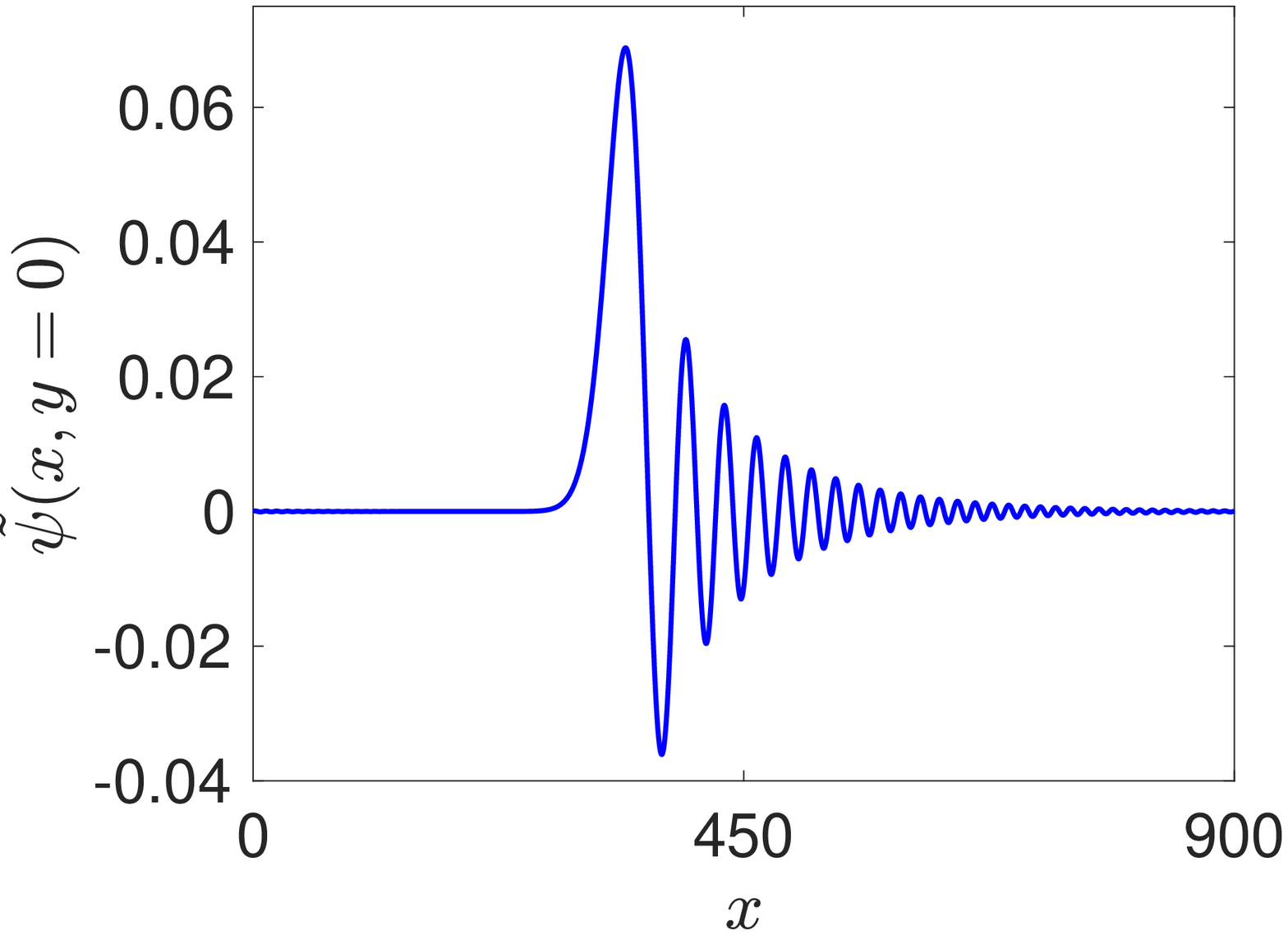}
	\end{tabular}
	\caption{Stream function $\tilde \psi(x,y,t=200)$ for the localized initial condition (\ref{e.A0}) with $A_0=0.11$ and $w=0.1$, for the symmetric mean flow depicted in Figure~\ref{fig.Usym}. Left: surface plot of the stream function. Right: Meridional cross-section at $y = 0$.}
	\label{fig.QGUsym}
\end{figure}

\begin{figure}[h]
	\begin{tabular}{ll}
		\includegraphics[width = 0.45\columnwidth]{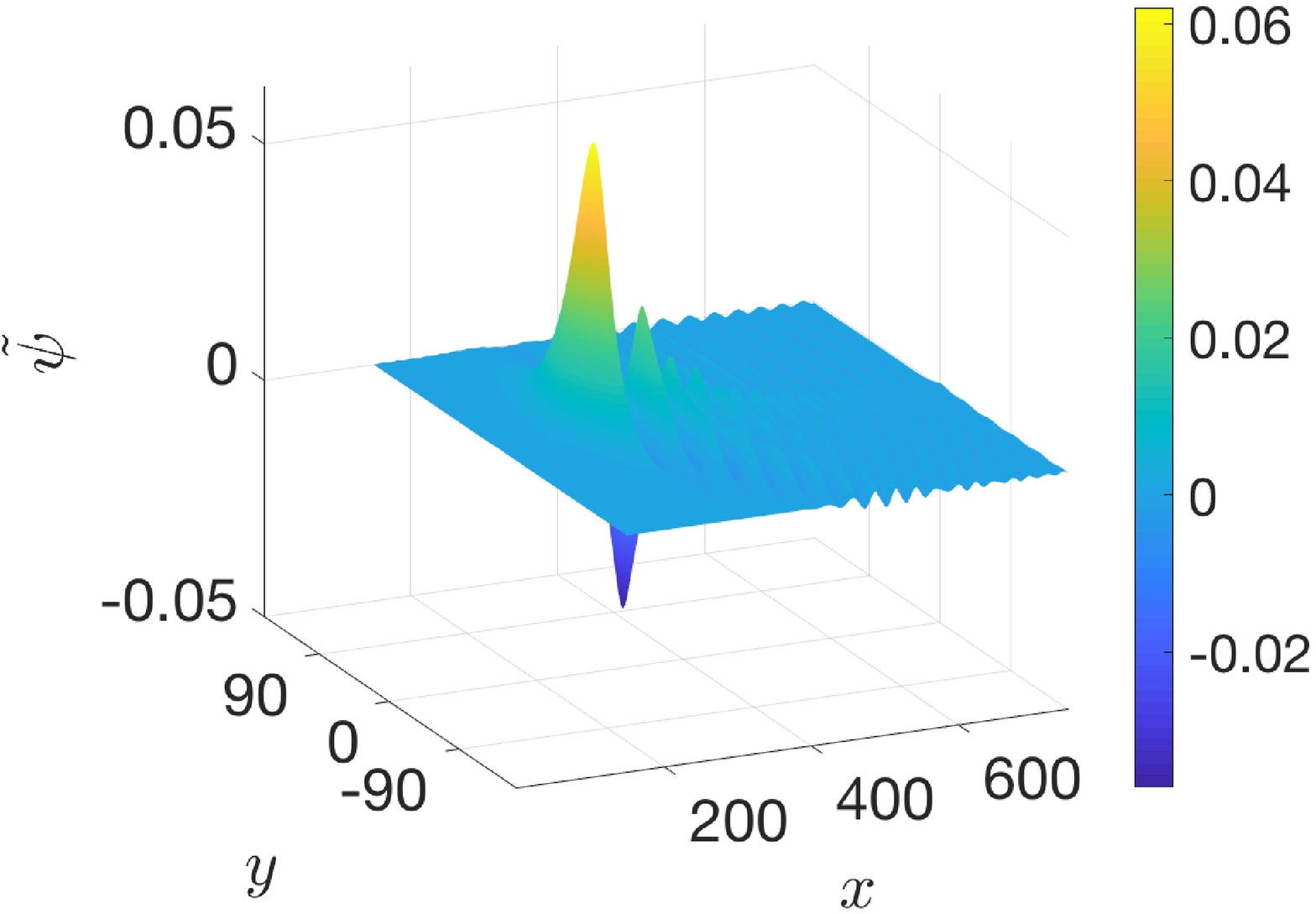}
		\includegraphics[width = 0.45\columnwidth]{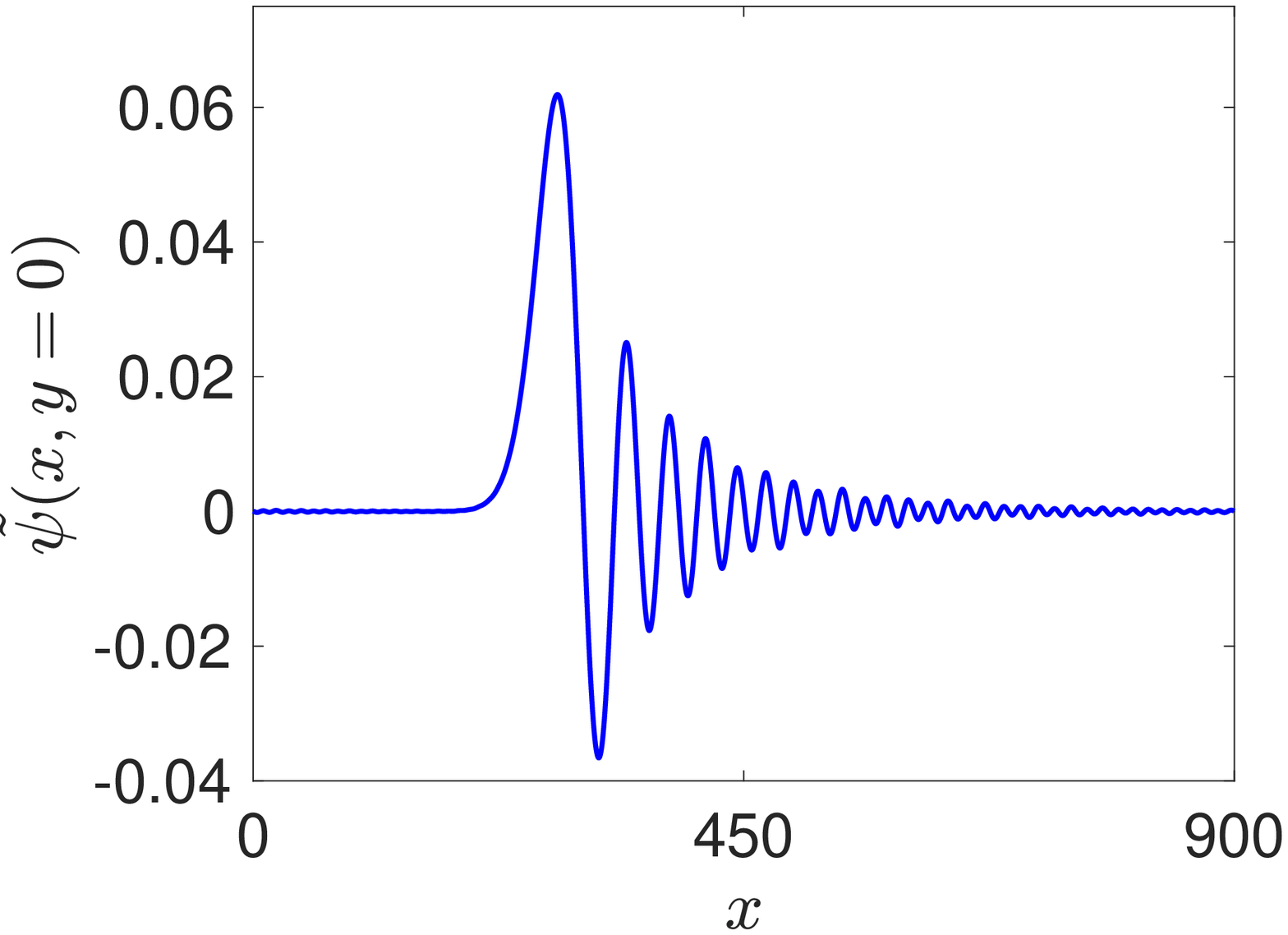}
	\end{tabular}
	\caption{Stream function $\tilde \psi(x,y,t=400)$ for the localized initial condition (\ref{e.A0}) with $A_0=0.11$ and $w=0.1$, for the asymmetric mean flow depicted in Figure~\ref{fig.Uasym}. Left: surface plot of the stream function. Right: Meridional cross-section at $y = 0$.}
	\label{fig.QGaUsym}
\end{figure}


\section{Discussion}
\label{sec.disc}

We have shown with a combination of rigorous analytical calculations and careful computational simulations 
that the dynamics of small localized large-scale perturbations in the quasi-geostrophic potential vorticity equation on unbounded domains is entirely linear. Moreover, we have shown that the dispersion is confined to
the zonal direction and does not spread meridionally. This renders the typical weakly
nonlinear wave equations such as the KdV and ZK equations and their higher-order modification obsolete in describing
coherent structures such as atmospheric blocking events, long lived eddies in the ocean
or coherent structures in the Jovian atmosphere such as the Great Red Spot. We remark though that the KdV and modified KdV equations can still be derived in meridionally bounded channels \cite{Redekopp77,GottwaldGrimshaw99,GottwaldGrimshaw99b}. 
Since the size of these meridionally bounded channels is small compared to the typical long wave length scale of the solitary wave, the ZK or the modified ZK equations, although suggested by the linear dispersion relation (\ref{e.dispZK}) in the long-wave limit, are therefore excluded as valid two-dimensional nonlinear wave models for large-scale slow localized structures.

\vspace{0.5cm}

\noindent
{\sc{Acknowledgments: }}
GAG would like to thank Oliver B\"uhler,
Daniel Daners, 
David Dritschel, Roger Grimshaw, Edgar Knobloch,
Victor Shrira and Vladimir Zeitlin for valuable discussions at various stages.
DEP acknowledges a financial support from the State task program in the sphere
of scientific activity of Ministry of Education and Science of the Russian Federation
(Task No. 5.5176.2017/8.9) and from the grant of President of Russian Federation
for the leading scientific schools (NSH-2685.2018.5).



\end{document}